\documentclass[12pt]{article}

\usepackage{amsmath}
\usepackage{times}  
\usepackage{helvet} 
\usepackage{courier}  
\usepackage{comment} 
\usepackage{graphicx} 
  
\usepackage{graphicx}  
\usepackage{bookmark}
\usepackage{amsthm}
\usepackage{amsfonts}
\usepackage{nicefrac}
\usepackage{times}
\usepackage{soul}
\usepackage{graphicx}
\usepackage{booktabs}
\newcommand{\citep}{\cite}
\usepackage{paralist}
\usepackage{xcolor}
\usepackage[ruled,vlined]{algorithm2e}
\newcommand{\xmark}{\text{\sffamily x}}
\newcommand{\cmark}{\checkmark}
\newcommand{\calR}{\mathcal{R}}

\newcommand{\shortcite}{\cite}

\newcommand{\nimrod}[1]{\textcolor{red}{Nimrod says: #1}}

\renewenvironment{description}[1][10pt]
  {\list{}{\labelwidth=10pt \leftmargin=#1
   }}
  {\endlist}
\newcommand{\insupp}[1]{}

\newtheorem{theorem}{Theorem}

\newtheorem{remark}{Remark}
\newtheorem{definition}{Definition}
\newtheorem{example}{Example}
\newtheorem{proposition}{Proposition}

\begin{document}

\title{Participatory Budgeting with Cumulative Votes}

\author{
Piotr Skowron\\
  University of Warsaw\\
  {\small \texttt{p.skowron@mimuw.edu.pl}}
\and 
Arkadii Slinko\\
   University of Auckland\\
  {\small \texttt{a.slinko@auckland.ac.nz}}
\and 
Stanisław Szufa\\
  Jagiellonian University\\
  {\small \texttt{stanislaw.szufa@uj.edu.pl}}
\and 
Nimrod Talmon\\
  Ben-Gurion University\\
  {\small \texttt{talmonn@bgu.ac.il}}
}

\maketitle

\begin{abstract}
In participatory budgeting we are given a set of projects---each with a cost, an available budget, and a set of voters who in some form express their preferences over the projects. The goal is to select---based on voter preferences---a subset of projects whose total cost does not exceed the budget. We propose several aggregation methods based on the idea of cumulative votes, e.g., for the setting when each voter is given one coin and she specifies how this coin should be split among the projects. We compare our aggregation methods based on (1) axiomatic properties, and (2) computer simulations.
We identify one method, Minimal Transfers over Costs, that demonstrates particularly desirable behavior. In particular, it significantly improves on existing methods, satisfies a strong notion of proportionality, and, thus, is promising to be used in practice.
\end{abstract}

\section{Introduction}

The idea of Participatory Budgeting (PB) was born in Brazil during the 1980s when political reformers explored ways to move beyond the political system associated with Brazil’s military dictatorship (1964-1985) based on exclusion and corruption \cite{wampler2012participatory}. They aimed to increase transparency of decision making of local bodies, hoping to enhance social justice and democracy~\cite{participatoryBudgeting}. 
At the first stage of PB, governments and civil society organizations identify the set of goals on the basis of certain principles without any reference to the budget. For example, many PB programs in Brazil use the “Quality of Life Index,” initially devised by the government in the city of Belo Horizonte. These goals are translated into particular projects that will be offered to the whole society to choose from. The second stage is actual voting in which voters express their views on the relative importance of the projects on offer. Thus, Participatory Budgeting is indeed a direct-democracy approach to budgeting. 

The most prominent applications of PB have been at the level of municipality, where a fraction of the municipal budget is decided through a residents-wide election. Other applications include, for example, \begin{inparaenum}[(1)]
\item An airline company deciding which movies to offer on its in-flight entertainment system~\cite{owaWinner}, where license costs for different movies can vary;
  and
\item A fans-owned soccer club wishing to sign deals with athletes.
\end{inparaenum}

Lately, PB has gained a considerable attention, and an increasing amount of funds are currently distributed this way\footnote{The web sites \url{http://www.participatorybudgeting.org} and \url{http://pbstanford.org} provide up-to-date data on the adoption of participatory budgeting in North America. The list of other cities that use PB can be found in \url{http://en.wikipedia.org/wiki/Participatory_budgeting_by_country}.}, thus the task of the research community is to gain understanding of its foundations, the properties of the existing procedures, and to suggest new more efficient procedures. The survey of some existing procedures is given in~\cite{azizsurvey}.\smallskip

 It is the second stage of PB which is of interest to us in this paper. 
We therefore assume that there is a set of projects, each with its own cost, a set of voters expressing their preferences over the projects, and a budget limit. The task of an aggregation procedure is to select a subset of the projects on offer reflecting voter preferences in the best possible way and whose total cost does not exceed the budget limit. (Formal definitions are given later). One of the motivating factors of our research is the close resemblance of participatory budgeting to multiwinner elections---in fact, when the costs of projects are all equal to one, the model for PB collapses to the model for multiwinner elections when a commitee of $k$ candidates is to be elected (so $k$ in this case is the budget). 

Even though the interest of ordinary people and political activists in PB is steadily increasing and its adoption has been steadily on the rise, the research on voting procedures for PB is still scarce. 
Specifically, with a few notable exceptions, the literature so far focused on procedures based on approval ballots, in which each voter submits a subset of projects that she finds acceptable, and, to a lesser extent, on ordinal ballots where each voter ranks the projects from the most to the least desirable one~\cite{azizsurvey}.
Here, we study a model where each of the $n$ voters is given virtual coins worth $L/n$, where $L$ is the budget limit, and she is asked to specify how this coin should be split among the projects. Distributing these coins between projects a voter can not only signal which projects are worthy of funding but also the intensity of preference. By analogy with the voting theory (see e.g., \citep{cole1949legal})  we call such ballots cumulative votes.

We suggest and study various aggregation methods for PB with cumulative votes and demonstrate that in the context of PB the use of cumulative ballots is a very natural choice.


Perhaps the main objection against using cumulative voting is concerned with the inherent trade-off between expressiveness and cognitive burden:
On the one hand, cumulative ballots allow for expressive votes, in particular they come with fine-grained information about intensities of voters' preferences. On the other hand, they require more effort from the voters. A first answer to this possible criticism is that in the context of participatory budgeting, the gains of expressiveness can outweigh the increasing cognitive burden; in particular, as we demonstrate later, most aggregation methods we discuss here for cumulative votes achieve a significantly better voters' satisfaction (a term we formally define later) than the analogous methods with simpler types of inputs.


Furthermore, as cumulative votes generalize both approval and ordinal ballots,\footnote{A ranking of the projects can be transformed into a cumulative ballot by applying a certain positional scoring function, such as Borda, to the ranking, and by normalizing the so-obtained scores.} our methods are applicable to these models as well. This means that voters are not required to use the full expressiveness of cumulative ballots, but can safely describe their preferences using approval or ordinal ones (which can then be translated, by the user interface or the algorithm itself, to cumulative ballots). The voters who wish to provide their preferences in a more expressive format can still do that.

We begin our investigation of aggregation methods for participatory budgeting with cumulative votes by considering three greedy aggregation methods, which are similar to top-$k$ multiwinner voting rules~\cite{mwchapter}; we refer to these rules as:
\begin{inparaenum}[(1)]
\item greedy by support;
\item greedy by excess; and
\item greedy by support over cost.
\end{inparaenum}
These methods are fairly natural, quite easy to explain, and computationally efficient. They have, however, certain drawbacks that we will demonstrate by showing their inability to satisfy few desirable axiomatic properties; in particular, these methods can completely disregard the preferences of minorities.
Thus, we continue our investigation and propose several adaptations of the Single-Transferable Vote (STV) rule to the PB setting. STV is a well-known ordinal-based multiwinner election rule that is widely used when proportional representation is sought, e.g., for parliamentary elections. 

Our adaptations of STV extend the original rule by:
\begin{inparaenum}[(1)]
\item allowing the voters to specify their preferences as cumulative ballots; and
\item specifying how to deal with projects with different costs.
\end{inparaenum}
The different adaptations vary in some key details. Here we informally describe one of them:
  Given a budget limit $L$, each of the $n$ voters gets a bag of $L/n$ coins which she can then distribute to the available projects as she wishes. Then, we look at the coin stash next to each project and, in each iteration, if there is a project which accumulated a sufficient funding, then we fund it, remove the cost of this project from its stash, and redistribute whatever is left in the stash to other projects in accord with the preferences of voters contributed to this project. Otherwise, we dismiss a project with the smallest stash next to it, redistribute its stash to other projects, and proceed iteratively.

We demonstrate the behavior of our rules through illustrative examples, showing their resilience to certain problems that the greedy rules suffer from. We further evaluate our STV adaptations by:
\begin{inparaenum}[(1)]
\item studying various axiomatic properties of these rules; and
\item reporting on simulations we performed.
\end{inparaenum}
We identify one rule, Minimal Transfers over Costs (MTC), that behaves particularly well with respect to our axiomatic properties, and that indeed produces proportional results. 

%

\section{Related Work}

There is a growing body of work on PB and related topics in the social choice literature (see, e.g.,~\cite{azizsurvey}). Firstly, they differ in the type of information solicited from the voters. Since most PB rules---in one way or another---adapt the existing voting rules, the input they solicit is usually in the traditional social choice formats: the voters are either asked to approve some of the projects (and hence disapprove the others) or to rank projects in accord to their desirability. In approval ballots voters are asked either to approve a fixed, say $k$, number of projects \cite{abpb} or to approve any set of projects whose total cost does not exceed the budget (knapsack voting) \cite{goel2019knapsack}. Goel et al also use pairwise comparisons as the input format: they ask voters to compare pairs of projects by the ratio between value and cost. These comparisons are
aggregated using variants of classic voting rules, including the Borda count and the Kemeny rules.

Klamler et al \cite{klamler2012committee} and Lu and Boutilier~\cite{lu2011budgeted} modified the idea of multiwinner voting rule which, given the rankings of the alternatives, selects a fixed number $k$ of alternatives given their costs and a budget. In these works rankings do not depend on the costs. An $\ell$-truncated ballots can also be used instead of full ballots \cite{bentert2020comparing}.

A few recent works have considered soliciting cardinal utilities of projects~\cite{fainMunSha18:core_and_pb,fluSkoTruWil:nash_complexity}, where the voters specify the utility of implementation of each project to them.
Benade et al \cite{benade2017preference} suggested their version of knapsack voting assuming that voters vote for the bundle that maximizes their utility. They also suggested to solicit voters' rankings of projects by their value (Value voting) or value-for-money (Value-for-money voting), and the so-called Threshold voting, where each voter specifies the subset of projects whose value is perceived to be above a predefined threshold.
Laruelle \cite{laruelle2020voting} assumes that these utilities are implicit and can be calculated from rankings. We will show that the known rules for PB with cardinal utilities sometimes exhibit undesirable behavior---in particular, they violate pretty basic proportionality axioms. In this paper we define rules which do not share these undesirable features.

The idea of proportional representation in PB transforms into a fairness issue and it attracted attention of researchers. It should not happen, for example, that no project is selected among those for which a large minority voted. In particular, the idea of justified representation for approval based budgeting rules was advocated in \cite{aussieone} and the idea of proportionality for solid coalitions in the ordinal setting in~\cite{aziz2019proportionally}.


Cumulative and cardinal ballots have been also considered in a divisible model of participatory budgeting, i.e., for aggregating divisions of available funds among existing projects, rather than to select projects for funding~\cite{FainGM16,dominiktruthful}.
%
%
A larger chunk of the literature on cumulative voting concerns single-winner elections; see, e.g.,~\cite{mills1968mathematics,bhagat1984cumulative,cole1949legal,vengroff2003electoral}. Most of this literature, yet, focuses only on practical and legal aspects of using cumulative voting.

 There is a large body of research on the famous voting rule Single Transferable Vote (STV) and there is some consensus among the social choice community that it is a particularly good rule for single-winner elections and, especially, multiwinner elections in cases where proportional representation is a desired property; see, e.g., \cite{tid-ric:j:stv,elk-fal-sko-sli:j:multiwinner-properties,mw2d}. This prompted us to adapt STV to PB via cumulative voting.
We also mention the work of Ford~\cite[Section 3.4]{ford2020liquid} that suggests using a cumulative version of STV for multiwinner elections.

The closest research project to ours is that described by \emph{accurate democracy}.\footnote{\url{https://www.accuratedemocracy.com/p_need2.htm}} There, an iterative procedure for PB that resembles STV is described. From the rather informal description that is given there we could identify a few key differences between the procedure described there and ours---in particular, the rule there uses both ordinal and cumulative ballots. Also, no axiomatic analysis of the rule is provided.

\section{Formal Model}

In our model there is a set of projects $P = \{p_1, \ldots, p_m\}$; the cost of a project $p \in P$ is a natural number, denoted~$c(p)$.
There is a set of $n$ voters $V = \{v_1, \ldots, v_n\}$, where voter~$v_j$ expresses her preferences over the projects by assigning a value $v_j(p)$ to each $p \in P$ such that $v_j(p) \geq 0$ and $\sum_{p \in P} v_j(p) = 1$ (notice that we refer to both the $j$th voter and her cumulative vote using the same symbol $v_j$); intuitively, the value of $v_j(p)$ is understood as the fraction of the funds owned by voter $v_j$ that the voter would like to assign to project~$p$.\footnote{An alternative approach would be to interpret $v_j(p)$ as the fraction of the available funds that voter $j$ thinks should be assigned to project $p$. These two interpretations are close and differ in whether we take the local or the global view on the voters' preferences. In this paper we take the local interpretation---we assume that the voters indirectly control the funds, and indicate how parts of funds that they control should be spread among the projects.}
We say that a voter $v_i$ \emph{supports} a project~$p$ if $v_j(p) > 0$.

The above notation naturally extends to sets. For each $B \subseteq P$ and each $v_j \in V$ we set $v_j(B) = \sum_{p \in B}v_j(p)$ and $c(B) = \sum_{p \in B} c(p)$.

A \emph{budgeting scenario} is a tuple $(P, V, c, L)$, where $P$, $V$, and $c$ are as defined above, and $L \in \mathbb{N}$ is a \emph{budget limit}.
An \emph{aggregation method} is a function (an algorithm) that, given a budgeting scenario, selects a \emph{bundle} of projects $B \subseteq P$ such that $c(B) \leq L$. 


\section{Greedy Rules}\label{sec:greedy_rules}

In this section we take, arguably, the most straightforward approach, and adapt known greedy algorithms for participatory budgeting~\cite{goel2019knapsack} to cumulative ballots. We start by describing the general class of greedy rules.

Let $f$ be a function that given a project $p$ returns a real value, called the \emph{priority} of~$p$. The greedy rule based on $f$ first ranks the projects in the descending order of their priorities, as given by $f$. Next, the rule iterates through the ranked list of the projects,\footnote{If $f$ is not injective, then we resolve ties using an arbitrary fixed tie-breaking rule.} in each iteration deciding whether the project at hand will or will not be selected. Let $L(t)$ denote the remaining budget in the $t$-th iteration of the procedure ($L(1) = L$). In the $t$-th iteration the rule examines project $p$: if $c(p) \leq L(t)$, then $p$ is selected, and the remaining budget is updated $L(t+1) = L(t) - c(p)$. Otherwise, $p$ is not selected, and $L(t + 1) = L(t)$. The following three aggregation methods are greedy rules: 

\begin{description}

\item[Greedy-by-Support (GS).] This is the greedy rule based on $f_{\mathrm{GS}}(p) = \sum_{j \in [n]} v_j(p)$.


\item[Greedy-by-Support-over-Cost (GSC).] It is based on  $f_{\mathrm{GSC}}(p) = (\nicefrac{1}{c(p)}) \cdot \sum_{j \in [n]} (v_j(p) \cdot (L / n))$.

\item[Greedy-by-Excess (GE).] Based on $f_{\mathrm{GE}}(p) = \sum_{j \in [n]} (v_j(p) \cdot (L / n)) - c(p)$.


\end{description}

\begin{remark}
We do not consider Greedy-by-Excess-over-Cost as it is equivalent to GSC.
\end{remark}

The first rule described above (GS) can be seen as an adaptation of Knapsack Voting~\cite{goel2019knapsack} to cumulative ballots. 
%
Yet, all three rules share a negative feature, namely that a significant part of the population of voters might be ignored when they split their votes on too many projects.

\begin{example}\label{greedy_split_votes}
Consider a set $P$ of 20 projects, all having the same cost equal to one, and a set of 100 voters with the following preferences: The first 60 voters consider the first 10 projects excellent and they all decide to assign the value $\nicefrac{1}{10}$ to each of them. The remaining 40 voters have quite opposite preferences---they decide to put the utility of $\nicefrac{1}{10}$ on each of the last 10 projects. The budget limit is $L = 10$. Here, GS, GSC, and GE would select the first 10 projects for funding, thus effectively ignoring the opinion of a large fraction of the society. \qed
\end{example}
\smallskip

Example~\ref{greedy_split_votes} also demonstrates the inherent difficulty of achieving proportionality with cumulative ballots. E.g., one aggregation method using cardinal utilities, which is known in the literature and considered proportional is based on the idea of maximizing the smoothed Nash welfare (SNW)~\cite{fainMunSha18:core_and_pb,fluSkoTruWil:nash_complexity}.
If we simply assume that cumulative ballots correspond exactly to cardinal utilities, then the rule would work as follows:
  For each budgeting scenario $(P, V, c, L)$, SNW would return a bundle $B$ which does not exceed the limit $L$ and which maximizes the product
  $\prod_{v_j \in V} \Big(v^*_j + \sum_{p \in B}v_i(p)\Big)$, where $v^*_j$ is the maximum value that voter $j$ can get in any feasible outcome.
SNW, when applied to the budgeting scenario from Example~\ref{greedy_split_votes} would inappropriately favor the majority of voters---it would select eight projects supported by 60\% of voters and 2 projects supported by 40\% of voters; it would share (yet to a lesser extent) the aforementioned negative feature of greedy rules.

In the next section we mainly aim at remedying the undesired behavior of GS, GSC, and GE by offering other aggregation methods which have a visible Single Transferable Vote (STV) flavor. Indeed, in the instance described in Example~\ref{greedy_split_votes}, our new methods would select $6$ projects supported by the group of $60\%$ voters.

\section{Cumulative Single Transferable Vote (CSTV) Rules}\label{section:cstv}

In this section we describe several adaptations of the Single Transferable Vote (STV) rule to the case of participatory budgeting with cumulative votes. We refer to these adaptations as Cumulative-STV, or, in short, as CSTV.
We first describe the general scheme, and later discuss several variants, which differ in certain key aspects regarding their specific operation.

All our variants of CSTV are based on the following compelling idea:
  Each from the $n$ voters should be able to decide where to allocate a $\nicefrac{1}{n}$-th fraction of the budget. Correspondingly, we say that a project $p$ is \emph{eligible for funding} if:
\begin{align*}
\mathrm{support}(p) = L \cdot \frac{\sum_{j=1}^n v_j(p)}{n} \geq c(p) \text{\ .}
\end{align*}

Notice that the total cost of the projects that are eligible for funding does not exceed the budget. At first it may seem reasonable to simply pick these projects and reject the others. This simple strategy, however, might often result in undesirable outcomes. For example, assume that there is a large number of excellent project proposals, each is liked by almost everyone, and that the voters decided to distribute their support roughly uniformly among the projects. In such a case it is possible that no project would be eligible for funding and such a simple rule would return the empty set (a similar behavior would be observed for the budgeting scenario given in Example~\ref{greedy_split_votes}). In order to deal with this and similar situations we allow the algorithm to perform certain transfers of the cumulative ballots of the voters between the projects. 

The idea of these transfers is as follows. Since voters do not have any coordination devices, they may allocate too much money for some projects. If the voter would be informed that her contribution is not needed for a certain project and that there will be enough funds without her (or that, even with her support, the project would not be funded), then she would divert her funds to other projects she liked. The CSTV rules take care of this oversupply of funds and redistribute voters' support on behalf of them.

\subsection{The General CSTV Rule}

Here we describe the general scheme. To make it a concrete rule, one has to specify the following subroutines:
\begin{inparaenum}[(1)] 
 \item project-to-fund selection procedure,
 \item excess redistribution procedure,
 \item no-eligible-project procedure, and
 \item inclusive maximality postprocedure.
\end{inparaenum}
We will discuss these subroutines in the subsequent part of this section.

The general scheme is as follows:
  Initialize $S = \emptyset$.
Loop over the following until a halting condition is met:
  If there are projects that are eligible for funding, choose one such project $p$ according to the ``project-to-fund selection procedure''. If the total value that the voters put on $p$ is strictly greater than the value needed for selecting the project (i.e., if $\mathrm{support}(p) > c(p)$), then for each voter $v_j$ with \mbox{$v_j(p) > 0$}, transfer a part of her initial support from  $p$ to other projects that $v_j$ had initially supported, so that $\mathrm{support}(p)$ is as close to $c(p)$ as possible. Such transfers are performed according to the ``excess redistribution procedure'', described below. Next, add $p$ to $S$, remove it from further consideration, reduce the available budget, $L:=L-c(p)$, and make the voters pay for $p$, i.e., for all $v_j$, set $v_j(p) = 0$.
  
  Else, that is, if there is no project eligible for funding, perform one of the following actions:
  \begin{inparaenum}[(1)] 
  \item select and eliminate a project $p$; transfer the values that the voters put on $p$ to other projects, or
  \item select a project $p$ and transfer values from other projects to $p$ so that it becomes eligible for funding. 
  \end{inparaenum}
   This step is performed according to the ``no-eligible-project'' policy. Finally, move to the  beginning of the loop.

After a halting condition is met, the remaining part of the budget might still be large enough to fund at least one additional project (in such a case, we say that the bundle of selected projects is not \emph{inclusive maximal}). If this is the case, we can run the ``inclusive maximality postprocedure'', or leave a part of the available budget unused.


        
      

Below we write the specific procedures which differentiate between the variants of CSTV we consider.

\subsection{Project-To-Fund Selection Procedure}

If there are multiple projects eligible for funding, we pick the one with the highest priority, using one of the following three priority functions:
\begin{inparaenum}[(1)]
\item $f_{\mathrm{GS}}$,
\item $f_{\mathrm{GSC}}$, or
\item $f_{\mathrm{GE}}$,
\end{inparaenum}
which we described in Section~\ref{sec:greedy_rules} in the context of greedy rules.

%



\subsection{Excess Redistribution Procedure}\label{sec:excess_redistribution}

In our further study we use the proportional strategy for redistributing the excess. To describe it formally, let $p$ denote the project currently selected for funding, and let $\mathrm{tran}(p)$ denote the set of voters who put a part of their support to~$p$ and also to some other not yet selected project:
\begin{align*}
\mathrm{tran}(p) = \{v_j \mid v_j(p) > 0 \ \text{and} \ \exists{p' \notin S\colon}\ v_j(p') > 0 \} \text{\ .}
\end{align*}

We make the payments proportional to the initial supports. We find $\gamma < 1$ such that:
\begin{align*}
\frac{\gamma L}{n} \sum_{v_j \in \mathrm{tran}(p)} v_j(p) \ \  + \ \  \frac{L}{n} \sum_{v_j \notin \mathrm{tran}(p)} v_j(p) =&\ c(p)\ \text{.}
\end{align*}
Intuitively, $\gamma$ is a factor such that if each voter $v_j \in \mathrm{tran}(p)$ scales her support for $p_i$ by $\gamma$, then $p$ will get exactly the support equal to its cost.
Next, for each $v_j \in \mathrm{tran}(p)$ we distribute $(1 - \gamma) \cdot v_j(p)$ among all not yet selected projects, proportionally to the initial supports that $v_j$ assigned to these projects.\footnote{
There are also other natural possibilities for redistributing the excess. For example, one could think of an additive version of proportional shares, which we term equal shares. The idea is to make the voters pay for the selected project as equal shares as possible. Formally, we find $\lambda$ such that 

\begin{align*}
\frac{L}{n} \cdot \sum_{v_j \in \mathrm{tran}(p)} \min(v_j(p), \lambda) + \frac{L}{n} \cdot \sum_{v_j \notin \mathrm{tran}(p)} v_j(p) = c(p) \text{\ ,}
\end{align*}
and for each voter $v_j \in \mathrm{tran}(p)$ with $v_j(p) > \lambda$ we distribute the surplus of the support $(v_j(p) - \lambda)$ among all not yet selected projects, proportionally to the initial supports that voter~$v_j$ assigned to these projects.

Another option would be to adapt an egalitarian criterion and to minimize the maximal transfer a voter must perform. Hereinafter we do not investigate these strategies in detail, yet we consider studying them an interesting direction for future work. 
}

\subsection{No-Eligible-Project Procedure}

We have two alternative procedures to apply when there is no project which is eligible to funding. These procedures are performed until certain project becomes eligible for funding.

\begin{description}

\item[Elimination-with-Transfers (EwT).] Here, we eliminate a project~$p$ with either the minimal $\mathrm{excess}(p) = \mathrm{support}(p) - c(p)$ or the minimal ratio $\mathrm{excess}(p) / c(p)$. If we chose the GE Project-to-fund selection procedure, then we do the former; if we chose the GSC Project-to-fund selection procedure then we do the latter. Once $p$ is chosen for elimination, then for each voter $v_j$ who put a part of their support on~$p$, we transfer this part to the other projects, proportionally to the initial supports that $v_j$ assigned to them. If $v_j$ put their support only on~$p$ then no transfers are made.
Notice that, if, at any time, there is a project that costs more than the total amount of money left, then it will be eventually eliminated and its money will be redistributed.

\item[Minimal-Transfers (MT).] We look for a project which may become eligible for funding if we transfer part of the support from other projects. Formally, we say that a project $p$ is \emph{eligible for funding by transfers} if 
\begin{align*}
\frac{L}{n} \cdot \sum_{j \colon v_j(p) > 0} \sum_{\ell=1}^m  v_j(p_\ell) \geq c(p)\ .
\end{align*}
That is, $p$ is eligible by transfers if its cost would be achieved provided the voters who initially put some positive utility value on $p$ would redirect all their remaining supports towards $p$.

If we choose the GE Project-to-fund selection procedure, then among all projects which are eligible by transfers, we select the one that minimizes the total amount of transfers that are required in order for it to became eligible; i.e., the project $p$ which (among those eligible by transfers) has maximal $\mathrm{excess}(p)$ (smallest in absolute value). If we choose GSC as the project-to-fund selection procedure, then among the projects that are eligible by transfers we pick the one, $p$, with the highest ratio $\mathrm{support}(p)/c(p)$.

We transfer the supports from other projects to $p$ so that $p$ reaches the eligibility threshold. We again follow the proportional strategy (cf. Section~\ref{sec:excess_redistribution}). In order to define proportional shares we iteratively do the following:
  First, we compute the ratio $r = \mathrm{support(p)}/c(p)$. If $r<1$, then every supporter $v_j$ of $p$ updates her votes. First the voter computes the desired support she should put to $p$: $v_j(p) := \min(\sum_{\ell =1}^m{v_j(p_\ell)}, \frac{v_j(p)}{r})$. Then, such a voter updates her votes, by proportionally transferring her support from other projects towards $p$.
We continue the procedure until $r=1$.

\end{description} 

\subsection{Inclusive Maximality Postprocedure}

We have two procedures to continue in cases when the algorithm halts but a part of the available budget is still unused.

\begin{description}

\item[Reverse Eliminations (RE).] We apply this procedure only when using Elimination-with-Transfers for selecting non-eligible projects. We iterate over the not-selected projects in the order reverse to the one in which they were eliminated. For each project, we check whether its cost does not exceed the available funds, and if so we fund it. This procedure is consistent with the logic of EwT, as EwT can be viewed as a procedure that creates a ranking of the projects: Whenever it adds a project $p$ for funding, it puts $p$ in the first available position in the ranking; when it eliminates $p$, it puts $p$ in the last available position. Thus, EwT with reverse eliminations is a greedy procedure that moves in the order consistent with the ranking returned by EwT.

\item[Acceptance of Undersupported Projects (AUP).] This procedure is allowed only when we use Minimal-Transfers as the no-eligible-project procedure. We proceed similarly as in MT, but this time we do not check the condition for eligibility by transfers. That is, if we use GE Project-to-fund selection procedure, then among not-yet selected projects that have their costs no-greater than the remaining budget, we pick the project $p$ which maximizes:
\begin{align*}
\frac{L}{n} \cdot \left(\sum_{j \colon v_j(p) > 0} \sum_{\ell=1}^m v_j(p_\ell)\right) - c(p) \text{\ ,}
\end{align*}
If we use GSC, then we choose~$p$ that maximizes:
\begin{align*}
\frac{\frac{L}{n} \cdot \sum_{j \colon v_j(p) > 0} \sum_{\ell=1}^m v_j(p_\ell)} {c(p)} \text{\ .}
\end{align*}

For each voter $v_j$ with $v_j(p) > 0$ we transfer all her support from other projects to $p$. We add $p$ to $S$ and repeat the procedure until no further project can be added.

\end{description}

\subsection{Selection of Variants}

The various design choices described above give rise to a number of aggregation methods, out of which we consider the following concrete CSTV aggregation methods:
  EwT (i.e., GE + EwT + RE),
  EwTC (i.e., GSC + EwT + RE),
  MT (i.e., GE + MT + AUP),
  MTC (i.e., GSC + MT + AUP); and out of the greedy rules we consider GS and GSC. We decided to chose those variants, since in our initial axiomatic and experimental analysis they gave the most promising results. 
  
%

\section{Axiomatic Properties}

In this section we compare our methods according to certain axiomatic properties.
We use these axioms to better understand the behavior of our rules. We concentrate on two types of axioms:
  (1) axioms that relate to monotonicity, as they are usually quite standard and provide general understanding of rules; and (2) axioms that relate to proportionality, as proportionality (e.g., taking care for minorities) is highly desired for participatory budgeting (see, e.g.,~\cite{aussieone}).
%
Our analysis is summarized in Table~\ref{table:axioms}.

\begin{table}[t]
	\centering
		\begin{tabular}{ c c c c c c c} 
			\toprule
			& GS & EwT & MT & GSC & EwTC & MTC \\
			\midrule
			Splitting monotonicity
			    & \xmark
			    & \cmark
			    & \cmark
			    & \cmark
			    & \xmark
			    & \xmark
			    \\
			Merging monotonicity
			    & \cmark
			    & \xmark
			    & \xmark
			    & \xmark
			    & \xmark
			    & \xmark
			    \\
			Support monotonicity
			    & \xmark
			    & \xmark
			    & \xmark
			    & \xmark
			    & \xmark
			    & \xmark
			    \\
			\midrule
			Weak-PR
			    & \xmark
			    & \cmark
			    & \cmark
			    & \cmark
			    & \cmark
			    & \cmark
			    \\
			PR
			    & \xmark
			    & \cmark
			    & \cmark
			    & \xmark
			    & \cmark
			    & \cmark
			    \\
			Strong-PR
			    & \xmark
			    & \xmark
			    & \cmark
			    & \xmark
			    & \xmark
			    & \cmark
			    \\
			\bottomrule
	\end{tabular}
	\caption{Axiomatic properties of GS, GSC, and the CSTV variants.}
	\label{table:axioms}
\end{table}

\subsection{Monotonicity Axioms}

In this section we consider three monotonicity axioms. The first two---splitting monotonicity and merging monotonicity---have been considered by Faliszewski and Talmon~\cite{abpb}, but in the context of approval-based preferences. 

Splitting monotonicity requires that, if a funded project $p$ is  split into several projects $P'$, according with its cost and the support it got from the voters, then at least one of the projects of $P'$ must be funded. Below we provide a formal definition.

\begin{definition}[Splitting monotonicity]
An aggregation method~$\mathcal{R}$ satisfies \emph{splitting monotonicity} if for each budgeting scenario $E = (P, V, c, L)$, for each funded project $p \in \mathcal{R}(E)$, and for each budgeting scenario $E'$ which is formed by splitting project $p$ into a set of projects $P'$ with the same cost $c(p) = c(P')$, and such that for each voter $v_i$ we have $v_i(P') = v_i(p)$, it holds that $\mathcal{R}(E') \cap P' \neq \emptyset$.
\end{definition}

Since splitting monotonicity seems a natural axiom 
it is quite surprising that only three of our CSTV rules satisfy it.

\begin{theorem}
  GS, EwTC, and MTC do not satisfy splitting monotonicity, while GSC, EwT and MTC satisfy the property.
\end{theorem}
\begin{proof}
\textbf{GS}: Consider the following budgeting instance with $1$ voter and $2$ projects, $p_1$ and $p_2$. Assume that $c(p_1) = c(p_2) = 2$ and that the budget is $L = 2$. The cumulative ballot of the single voter $v$ is $v(p_1) = 0.6$ and $v(p_2) = 0.4$. GS would select $p_1$. Now, assume that we split $p_1$ into two projects, $p_a$ and $p_b$, such that $c(p_a) = c(p_b) = 1$ and such that $v(p_a) = v(p_b) = 0.3$. Now, GS would select $p_2$, failing splitting monotonicity.

\textbf{GSC}: First, observe that in order to prove that a rule satisfies splitting monotonicity it is sufficient to consider cases where the project that is to be split is divided into two parts. Indeed, if the project were divided into more than two parts, one could use the reasoning for splitting into two projects recursively.

Consider a project $p$ that is about to be split into $p_a$ and~$p_b$. Observe that either 
\begin{align*}
&\frac{\mathrm{support}(p_a)}{c(p_a)} \geq \frac{\mathrm{support}(p)}{c(p)} \text{\ , or} \\
&\frac{\mathrm{support}(p_b)}{c(p_b)} \geq \frac{\mathrm{support}(p)}{c(p)}
\end{align*}
Indeed, if none of the above inequalities held, then we would have:
\begin{align*}
1 = \frac{\mathrm{support}(p_a) + \mathrm{support}(p_b)}{\mathrm{support}(p)} < \frac{c(p_a)}{c(p)} + \frac{c(p_b)}{c(p)} = 1 \text{\ ,}
\end{align*}
a contradiction. W.l.o.g., let us assume that:
\begin{align*}
\frac{\mathrm{support}(p_a)}{c(p_a)} \geq \frac{\mathrm{support}(p)}{c(p)}\  \text{.}
\end{align*}

Assume that $p$ is selected by the rule. Then, after split, $p_a$ would be selected in the same moment as $p$ was selected, or before. This proves that GSC satisfies splitting monotonicity.

\textbf{EwT}: Consider a project $p$ that is selected for funding and assume $p$ is split into two projects, $p_a$ and $p_b$. Note that whenever a voter $v$ transfers some value to $p$ (either as a result of redistributing the excess or as a result of being removed), then after the split, $v$ transfers to $p_a$ and $p_b$ the same total value as she transferred to $p$---this is because we use the proportional strategy of redistributing values. Furthermore, $v$ transfers to the other projects the same value as before the split.

Consider the steps of the algorithm before $p$, $p_a$, or $p_b$ is removed or selected. In each such step we have that:
\begin{align*}
    \mathrm{excess}(p_a) + \mathrm{excess}(p_b) = \mathrm{excess}(p) \text{ .}
\end{align*}
First, consider the case where neither $p_a$ or $p_b$ is removed before $p$ is selected or removed. If $p$ is selected (in such a time moment its support is greater than or equal to its cost), then either $p_a$ or $p_b$ is eligible for funding, and so it will be eventually selected. If $p$ is removed (which means it is added in the RE phase of the algorithm), then by our assumption $p_a$ and $p_b$ will be removed after $p$, and so they both will be added for funding in the RE phase.

Second, consider the case when $p_a$ or $p_b$---say $p_a$---is removed before $p$ is selected or removed. At the time $p_a$ is removed we have that:
\begin{align*}
\mathrm{excess}(p_a) \leq \mathrm{excess}(p) \text{ .}
\end{align*}
In that moment we have:
\begin{align*}
\mathrm{excess}(p_b) \geq 0 \text{ .}
\end{align*}
Hence, $p_b$ will be eventually selected. As before, the reasoning can be recursively applied to the case when $p$ is split to any set of projects $P'$.

\textbf{EwTC}: Consider the following instance with 3 projects, $p_1, p_2$, and $p_3$. Their costs are $c(p_1) = 199$, $c(p_2) = 102$, and $c(p_3) = 200$, and the budget is $L = 200$. Let us fix a small constant $\epsilon$. There are 200 voters with the cumulative ballots given in the following table:

\begin{center}
   \begin{tabular*}{0.9\linewidth}{@{\extracolsep{\fill}}@{}lccc@{}}
    \toprule
      \# votes  & $p_1$ & $p_2$ & $p_3$ \\ 
    \midrule
    $140$  & $\nicefrac{9}{14}$ & $\nicefrac{5}{14}$ & 0 \\ 
    $60$  & $\nicefrac{1}{6} - \epsilon$ & $\epsilon$ & $\nicefrac{5}{6}$ \\ 
    \bottomrule
   \end{tabular*}
\end{center}

The support of $p_1, p_2$, and $p_3$ equals respectively, $100 - 60\epsilon$, $50 + 60\epsilon$, and $50$. Thus, $p_3$ will be eliminated first. The last 60 voters will transfer (almost) the entire support to $p_1$. Thus, $p_2$ will be eliminated next, and so $p_1$ will be selected. 

Now, assume that $p_1$ is split into $p_a$ and $p_b$ such that $c(p_a) = 100$, and $c(p_b) = 99$. The voters' preferences look as follows:

\begin{center}
   \begin{tabular*}{0.9\linewidth}{@{\extracolsep{\fill}}@{}lcccc@{}}
    \toprule
      \# votes  & $p_a$ & $p_b$ & $p_2$ & $p_3$ \\ 
    \midrule
    $140$  & 0 & $\nicefrac{9}{14}$ & $\nicefrac{5}{14}$ & 0 \\ 
    $60$  & $\nicefrac{1}{6} - \epsilon$ & 0 & $\epsilon$ & $\nicefrac{5}{6}$ \\
    \bottomrule
   \end{tabular*}
\end{center}
Then, $p_a$ is eliminated first, and the last 60 voters transfer their almost entire support from $p_a$ to $p_3$. Next, $p_3$ is eliminated, and the whole value of the last 60 voters is transferred to $p_2$. Thus, the total support of $p_2$ becomes 110. In the last step $p_2$ is selected, and neither $p_a$ nor $p_b$ fit within the remaining budget. 

\textbf{MT}: Here the reasoning is similar to the case of EwT. First, we observe that the sum of transfers to each project stays the same after the split. If at some point $p$ was eligible for funding, then at this point $p_a$ or $p_b$ was as well, and so one of these projects would be selected. Thus, from now on, let us assume that the excess of $p$ was always negative. Observe that either the excess of $p_a$ or $p_b$ becomes positive, in which case the project would be selected, or the excesses of both $p_a$ and $p_b$ are greater than that of $p$. Furthermore, if $p$ is eligible by transfers, then $p_a$ or $p_b$ is as well. Thus, $p_a$ or $p_b$ will be selected at most at the time when $p$ was.

\textbf{MTC}: Consider the instance with 2 projects, $p_1$ and $p_2$, with costs equal to $c(p_1) = 150$ and $c(p_2) = 151$. The budget is $L = 200$. There are 200 voters with the following preferences:

\begin{center}
   \begin{tabular*}{0.9\linewidth}{@{\extracolsep{\fill}}@{}lcc@{}}
    \toprule
      \# votes  & $p_1$ & $p_2$ \\ 
    \midrule
    $150$  & $\nicefrac{1}{3}$ & $\nicefrac{2}{3}$  \\ 
    $50$   & $1$ & $0$  \\ 
    \bottomrule
   \end{tabular*}
\end{center}

In this example, the supports of the two projects are the same, they both are eligible by transfers, thus MTC selects the cheaper one---i.e., $p_1$.

Now, assume that $p_1$ is split into $p_a$ and $p_b$ such that $c(p_a) = 51$, and $p_b = 99$. The voters' preferences look as follows:

\begin{center}
   \begin{tabular*}{0.9\linewidth}{@{\extracolsep{\fill}}@{}lccc@{}}
    \toprule
      \# votes  & $p_a$ & $p_b$ & $p_2$ \\ 
    \midrule
    $150$  & 0 & $\nicefrac{1}{3}$ & $\nicefrac{2}{3}$ \\ 
    $50$   & $1$ & 0 & 0 \\
    \bottomrule
   \end{tabular*}
\end{center}
After the split, $p_a$ is not eligible by transfers; $p_2$ will be selected, leaving no room for $p_a$ nor $p_b$.
\end{proof}

The next axiom is analogous to the previous one, yet it describes merges among the projects rather than splits.  

\begin{definition}[Merging monotonicity]
An aggregation method~$\mathcal{R}$ satisfies \emph{merging monotonicity} if for each budgeting scenario $E = (P, V, c, L)$,  each $B' \subseteq \mathcal{R}(E)$, and for each scenario $E' = (P \setminus B' \cup \{b'\}, V, c', L)$ such that $b'$ is a new project which costs $c(B')$ and such that for each voter $v_i$ we have that $v_i(b') = \sum_{b \in B'} v_i(b)$, it holds that $b' \in \mathcal{R}(E')$.
\end{definition}

That is, merging monotonicity requires that a project formed by merging a number of funded projects is funded.

\begin{theorem}
  GS  satisfies merging monotonicity,
  while EwT, EwTC, MT, MTC, GSC fail merging monotonicity.
\end{theorem}

\begin{proof}
\textbf{MT, MTC}: Consider an instance with 5 projects, $p$, $q$, $r$, $s$, and $z$ with all the costs equal to 10. The first 10 voters put the value $0.5 + \epsilon$ to $p$ and $0.5 - \epsilon$ to $q$. The next 9 voters put $0.5 + \epsilon$ to $r$ and $0.5 - \epsilon$ to $s$. The last voter puts 1 to $z$. There are $n = 10 + 9 + 1 = 20$ voters; the budget is $L = 20$.

Since $p$ and $q$ are eligible by transfers, both MT and MTC will select $p$ first. For that, most of $q$'s money will be transferred to $p$ and as a result both rules will select $r$ in the second iteration. 
Now, assume that $p$ and $r$ are merged into $x$. The merged project is no longer eligible by transfers, but $q$ is still. It will be selected, and there will be no money left in the budget to buy~$x$.

\textbf{EwT, EwTC}: Consider the following instance with 5 projects, $p, q, r, s$, and $z$. The costs of $p$, $s$, and $z$ are equal to 30; the cost of $r$ is equal to $40$. If we consider EwT, then we set the cost of $q$ to 35; if we consider EwTC, we set $c(q) = 30$. 
The first 30 voters assign value 1 to $p$. The next voter assigns value $1 -\epsilon$ to $p$ and $\epsilon$ to $q$. The next 15 voters assign $1-\epsilon$ to $q$ and $\epsilon$ to $r$; the next 40 voters assign $0.5$ to $r$ and $0.5$ to $s$; the remaining $9$ voters assign 1 to $z$. The number of voters is $95$ and the budget is $L = 95$. The cumulative ballots of the voters are summarized in the table below.

\begin{center}
   \begin{tabular*}{0.9\linewidth}{@{\extracolsep{\fill}}@{}lccccc@{}}
    \toprule
      \# votes  & $p$ & $q$ & $r$ & $s$ & $z$\\ 
    \midrule
    $30$  & $1$ & $0$ & $0$ & $0$ & $0$ \\ 
    $1$   & $1-\epsilon$ & $\epsilon$ & 0 & $0$ & $0$ \\
    $15$  & 0 & $1-\epsilon$ & $\epsilon$ & 0 & $0$  \\
    $40$  & 0 & 0 & $0.5$ & $0.5$ & $0$  \\
    $9$  & 0 & $0$ & 0 & $0$ & $1$  \\
    \bottomrule
   \end{tabular*}
\end{center}

Here, both EwT and EwTC will select $p$ first. The value of the 31st voter will be transferred to $q$. Then, the support for $q, r$, and $s$ will be $16 - 15\epsilon$, $20 + 15\epsilon$, and $20$, respectively. Thus, $z$ and $r$ will be eliminated next; there will be enough money to accommodate the remaining two projects, thus $p, q$ and $s$ will be selected.   

Now assume that $p$ and $s$ are merged into $x$: it does not reach the threshold. Indeed, the support of projects $x, q$ and $r$ will be $51-\epsilon$, $15 - 14\epsilon$, and $20 + 15\epsilon$, respectively. Thus, $z$ and $q$ will be eliminated next. The total value of $15 - 15\epsilon$ will be transferred from $q$ to $r$, raising its support to $35$. Thus, $x$ will be eliminated next, and $r$ will be chosen, leaving no money in the budget for~$x$.

\textbf{GSC}: Consider an instance with 3 projects $p_1$, $p_2$, and $p_3$ with the costs equal to, respectively, 5, 10, and 5. The budget is $L = 10$. There are 10 identical voters, who assign value 0.35 to $p_1$, $0.6$ to $p_2$ and $0.05$ to $p_3$. GSC selects $p_1$ first, and then there will be no money for $p_2$. Consequently, $p_1$ and $p_3$ will be selected. On the other hand, if we merge $p_1$ and $p_3$, then the rule will select $p_2$.

\textbf{GS}: Observe that merging two projects does not affect the order of consideration of the other projects apart from the merged ones. Assume we merged $p_a$ and $p_b$ into $p$ and that $p_b$ was considered after $p_a$ by the rule. Thus, $p$ will be considered in the same or earlier iteration as $p_a$, and there will be enough money to accommodate $p$ (since at this moment, before merge, there was enough money to add $p_a$ and $p_b$). 
\end{proof}


Our next axiom, support monotonicity, requires that moving more support to a funded project does not hurt it. In the formulation below $\triangle$ stands for symmetric difference.

\begin{definition}[Support monotonicity]
An aggregation method~$\mathcal{R}$ satisfies \emph{support monotonicity} if for each budgeting scenario $E = (P, V, c, L)$, each project $p \in \mathcal{R}(E)$, and each budgeting scenario $E' = (P, V', c, L)$ such that $|V \triangle V'| = 1$ and for the single voter $v \in V \triangle V'$ and the single voter $v' \in V' \triangle V$ it holds that (1) $v'(p) > v(p)$ and (2) for each $p' \neq p$, $v'(p') \leq v(p')$, then $p \in \mathcal{R}(E')$. 
\end{definition}

\begin{theorem}\label{thm:support_monotonicity}
  GS, GSC, EwT, EwTC, MT, and MTC fail support monotonicity.
\end{theorem}

\begin{proof}
\textbf{EwT, EwTC}: Consider the following instance with 3 projects, $p, q, r$. The costs of all projects are equal to 10. The first 8 voters put $0.5 - \epsilon$ to $p$, $0.25$ to $q$ and $0.25 + \epsilon$ to $r$. The next 2 voters put 1 to $q$ and the next $2$ to put $1$ to $r$. There are 12 voters and the budget is $L = 12$.

Here, $p$ is eliminated first, $q$ is eliminated second, and so $r$ is selected for funding.

Now, assume that all the voters move $2\epsilon$ from $q$ to $r$. Now, $q$ is eliminated first and its value is redistributed among $p$ and $r$ in proportions $2:1$. Thus, $r$ is eliminated second, and so $p$ is the project selected for funding.

\textbf{MT}: Assume we have 3 projects, $p, q$ and $r$, with the costs equal to $4$, $4$ and $8$, respectively. There are 10 identical voters who put utility $0.25+\epsilon$ on $p$,  $0.1$ on $q$, and $0.65-\epsilon$ on $r$. The budget is $L = 10$.

MT will select $p$ first. Since there will be no money left for $r$, $q$ will be selected as a second project for funding. Now, assume that each voter moves $2\epsilon$ from $p$ to $q$. Now, MT selects $r$ first and there will be no money left for $q$.

\textbf{GS}: The instance is similar as for MT. The projects have the same costs but the voters put $\nicefrac{1}{3}+2\epsilon$ on $p$, $\nicefrac{1}{3} + \epsilon$ on $r$ and $\nicefrac{1}{3}-3\epsilon$ on $q$. Here GS will select $p$ and $q$. Now assume that the voters $2\epsilon$ from $p$ to $q$. After such a change GS will select $r$.

\textbf{MTC, GSC}: The constructions here are analogous to MT and GS.
\end{proof}

Theorem~\ref{thm:support_monotonicity} suggest the following interesting open question: does there exist a proportional rule for PB that satisfies our three monotonicity axioms, in particular support monotonicity? Does there exist an impossibility result suggesting that the axioms we consider in the paper are incompatible?

\subsection{Proportional Representation}

Next we consider proportionality as it is usually desired for PB applications (see, e.g.,~\cite{aussieone}); in particular, we see the lack of proportionality of the greedy rules -- as we show below -- as their major drawback.
Specifically, we introduce three axiomatic properties, of Weak-PR, PR, and Strong-PR. These axioms are new to the paper, but similar properties have been considered in the context of participatory budgeting for different types of voters' preferences~\cite{aussieone,fainMunSha18:core_and_pb}.


\begin{definition}[Weak Proportional Representation]
An aggregation method~$\mathcal{R}$ satisfies \emph{Weak Proportional  Representation (Weak-PR)} if for each budgeting scenario $E = (P, V, c, L)$, for each $\ell \in [L]$, each set $V' \subseteq V$ of voters with $|V'| \geq \ell n / L$, and each set $P' \subseteq P$ of projects with $c(P') \leq \ell$, there exist a scenario $E'$ which differs from $E$ only in the votes of the voters from $V'$, such that $P' \subseteq \mathcal{R}(E')$.
\end{definition}

%

\begin{theorem}\label{theorem:weakpr}
  GSC satisfies Weak-PR but GS fails it.
\end{theorem}

\begin{proof}
\textbf{GS}:
Consider a scenario with $P = \{a, b\}$, $c(a) = 1$, $c(b) = 3$, $L = 3$, and voters $v_1$, $v_2$, and $v_3$, where voter $v_1$ supports only $a$, and $v_2$ and $v_3$ support only project $b$.
For $\ell = 1$, voter $v_1$ acts as a set $V'$ of $|V'| \geq \ell n / L = 1$ and $P' = \{a\}$ acts as a set of projects with $c(P') \leq \ell$; GS, however, selects only project $b$, as project $b$ has higher support and, after $b$ is funded, no funds are left to fund project $a$.

\textbf{GSC}:
Let $\ell \in [L]$. Consider a set $V'$ of voters with $|V'| \geq \ell n / L$ and a set $P' \subseteq P$ of projects with $c(P') \leq \ell$. Set the vote of each $v' \in V'$ to support only the projects in $P'$, proportionally:
  i.e., for each $p' \in P'$ set $v'(p') = c(p') / c(P')$ and $v'(p) = 0$ for each $p \notin P'$. 
Now, the sum of support each project $p' \in P'$ gets from the voters is at least $|V'| \cdot c(p') / c(P')$, as it gets this amount already from the voters in $V'$. 
As $|V'| \geq \ell n / L$, we have that the sum of support of each project $p' \in P'$ is at least $\ell n / L \cdot c(p') / c(P')$; furthermore, since $c(P') \leq \ell$, it follows that this sum of support is at least $n / L \cdot c(p')$.

GSC ranks projects according to their sum of support over their cost, so the ``support over cost'' value of each $p' \in P'$ is at least $n / L$.
We wish to upper bound the total cost of projects $p \notin P'$ which get a ``support over cost'' value greater than $n / L$:
  The proof will follow by showing that the total cost of such projects is at most $L - \ell$,
  because then, GSC would fund all projects $p' \in P'$.
To show this, assume otherwise, that the total cost of projects $p \notin P'$ with ``support over cost'' value greater than $n / L$ is more than $L - \ell$, call the set of these projects $S$.

Observe that the number of voters $v \notin V'$ is at most $n - \ell n / L$. Thus, the total support divided by the total cost must be lower than 
$\frac{n - \frac{\ell n}{L}}{L - \ell} = \frac{n}{L}$; hence, contradiction.
\end{proof}

Intuitively, the three axioms that we consider in this section differ in how much synchronization is needed among the members of a group of voters in order ensure that these voters will be able to decide about a certain fraction of the budget.
Weak-PR ensures that a group of at least $\ell n / L$ can (by coordinating) make a set of projects $P'$ of total cost $\ell$ funded; PR merely requires that such voters support the same set $P'$.

\begin{definition}[Proportional Representation]
An aggregation method~$\mathcal{R}$ satisfies \emph{Proportional Representation (PR)} if for each budgeting scenario $E = (V, P, c, L)$,  each $\ell \in [L]$, each $V' \subseteq V$ with $|V'| \geq \ell n / L$, and each set $P' \subseteq P$ of projects with $c(P') \leq \ell$, it holds that:
  If all voters $v' \in V'$ support all projects in $P'$, and no other projects,
  then $P' \subseteq \mathcal{R}(E)$.
\end{definition}

\begin{theorem}\label{theorem:pr}
  EwT and EwTC satisfy PR but GSC fails it.
\end{theorem}

\begin{proof}
\textbf{GSC}:
Intuitively, GSC fails PR because after one project in $|P'|$ is funded, its excess gets lost, which might cause the other projects in $P'$ not funded.

More formally,
consider a budgeting scenario with $P = \{a, b, c\}$ where $c(a) = c(b) = 1$ and $c(c) = 3$, the budget limit $L = 4$, and voters $v_1$ and $v_2$, where voter $v_1$ assigns to project $a$ value $1 - \epsilon$ and to project $b$ value $\epsilon$, and voter $v_2$ assigns to project $c$ value $1$.

According to PR, with $V' = \{v_1\}$, $\ell = 2$, and $P' = \{a, b\}$, we have that indeed both $a$ and $b$ shall be funded. GSC, however, will choose the bundle $\{a, c\}$.

\textbf{EwT, EwTC}:
Fix $\ell \in [L]$. Let $P'$ be a set of projects with $c(P') \leq \ell$. Let $V'$ be a group of voters who all support all projects from $P'$ and no other projects; assume $|V'| \geq \ell n / L$.

Recall that we define the support of a project $p$ as:
\begin{align*}
\mathrm{support}(p) = L \cdot \frac{\sum_{j=1}^n v_j(p)}{n} \text{\ .}
\end{align*}
Let $S_i$ be the set of projects picked by the rule (EwT or EwTC) up to the $i$th iteration, inclusive.  
We prove the following invariant: In the $i$th step the total support the voters from $V'$ assign to the projects from $P' \setminus S_i$ equals at least $\ell - c(P' \cap S_i)$ and no project from $P'$ has been eliminated by the rule. The invariant is clearly true when the rule begins. Now, assume it is true after the $i$th iteration, and we will show that it must hold after the $(i+1)$th iteration as well. Observe that in the $(i+1)$th iteration the total support of the candidates from $P'$ equals at least:
\begin{align*}
\ell - c(P' \cap S_i) \geq c(P') - c(P' \cap S_i) = c(P' \setminus S_i) \text{ .}
\end{align*}
Thus, there must exists at least one project, support of which exceeds the cost, thus no project from $P'$ can be eliminated. Furthermore, if a project $p' \in P'$ is selected, then the amount of support that the voters from $V'$ assign to the projects from $P'$ will decrease by $c(p')$ (the exceed will be transferred only to the projects from $P'$, unless all of them are already selected). This proves the invariant.
Since no project from $P'$ will be eliminated, all of them will be picked by the rule.
\end{proof}

Below is our strongest proportionality axiom, in which we relax the requirement that $c(P') < \ell$.  According to Strong-PR, the voters in groups do not have to strongly synchronize to get projects that they like: intuitively, they only need to agree on the set of those projects that get a positive support. 

\begin{definition}[Strong Proportional Representation]
An aggregation method~$\mathcal{R}$ satisfies \emph{Strong Proportional Representation (Strong-PR)} if for each scenario $E = (P, V, c, L)$, each $\ell \in [L]$, each $V' \subseteq V$ with $|V'| \geq \ell n / L$, and each $P' \subseteq P$, it holds that:
  If all voters $v' \in V'$ support all projects in $P'$, and not any other project,
  then either $P' \subseteq \mathcal{R}(E)$ or for each $p \in P' \setminus \mathcal{R}(E)$ we have that $c(p) + c(P' \cap \mathcal{R}(E)) > \ell$.
\end{definition}

\begin{theorem}\label{theorem:strongpr}
  MT and MTC satisfy Strong-PR but EwT and EwTC fail it.
\end{theorem}
\begin{proof}
\textbf{EwT, EwTC}:
Consider an instance with $2$ voters, $v_1, v_2$ and $3$ projects, $p_1, p_2$, and $q$, such that $c(p_1) = 5$, $c(p_2) = 7$, and $c(q) = 6$. The budget limit is $L = 10$.  Voter $v_1$ puts $\epsilon$ to $p_1$ and $1-\epsilon$ to $p_2$. Voter $v_2$ puts $1$ to $q$. The support of the projects $p_1, p_2$, and $q$ will be respectively, $5\epsilon$, $5 - 5\epsilon$, and $5$. Thus, both EwT and EwTC will eliminate $p_1$ first, $p_2$ second, and $q$ last. Consequently, only $q$ will be selected while Strong-PJR requires selecting $p_1$.   

\textbf{MT, MTC}:
Fix $\ell \in [L]$, and consider a group of voters $V'$ with $|V'| \geq \ell n / L$, and a set $P' \subseteq P$ of projects, which are supported by all the voters from $V'$; furthermore, assume the voters from $V'$ do not support any other projects.

Let $S_i$ denote the set of projects selected by the rule (MT, or MTC) up to the $i$th iteration, inclusive.  
First, we observe that the following invariant holds: In each iteration $i$ the total support that the voters from $V'$ assign to the projects from $P' \setminus S_i$ equals at least $\ell - c(S_i \cap P')$. Indeed, the invariant is initially true (since $c(S_0 \cap P') = 0$ and $L/n \cdot |V'| \geq \ell$), and each time a project $p \in P'$ is selected, the total support that voters from $V'$ assign to projects is decreased by at most $c(p)$. Furthermore, the excess of the value that the voters from $V'$ assign to $p$ is redistributed only among the projects from $P'$. 

Now, for the sake of contradiction, assume there exists a project $p' \in P'$ such that  $c(p') + c(P' \cap \mathcal{R}(E)) \leq \ell$ and that has not been selected. Let $j$ be the last iteration before the rule reached the ``Inclusive Maximality Postprocedure'' phase (possibly $j$ is the last iteration of the rule). Then, clearly  $c(p') + c(P' \cap S_j) \leq \ell$. By our invariant, we get that in the $j$th iteration the total support the voters from $V'$ assign to the projects from $P' \setminus S_j$ equals at least:
\begin{align*}
\ell - c(S_j \cap P') \geq c(p') \text{ .}
\end{align*}
Furthermore, all the voters from $V'$ support $p'$, thus $p'$ is eligible by transfers. Consequently, the rule cannot stop nor reach the ``Inclusive Maximality Postprocedure'' phase. This gives a contradiction and completes the proof.
\end{proof}

\begin{proposition}
 Each aggregation rule that satisfies Strong-PR also satisfies PR. Each aggregation rule that satisfies PR also satisfies Weak-PR.
\end{proposition}

\begin{proof}
\textbf{Strong-PR $\to$ PR}:
Let $\calR$ be an aggregation method satisfying Strong-PR and, counterpositively, assume that $\calR$ fails PR. Since $\calR$ fails PR, then, by definition, there exists a budgeting scenario $E = (P, V, c, L)$, some $\ell \in [L]$, a set $V' \subseteq V$ of voters with $|V'| \geq \ell n / L$, a set $P' \subseteq P$ of projects with $c(P') \leq \ell$, and a project $p \in P'$ such that all voters $v' \in V'$ support all projects in $P'$, and not any other project, but $p \notin \mathcal{R}(E)$.

Now, since $\calR$ satisfies Strong-PR, then, by definition, for these specific $E$, $\ell$, $V'$, $P'$, and $p$, it holds that either:
\begin{enumerate}
\item
$P' \subseteq \mathcal{R}(E)$,
\item
for each $p \in P' \setminus \mathcal{R}(E)$ we have that $c(p) + c(P' \cap \mathcal{R}(E)) > \ell$.
\end{enumerate}

But, for these specific $E$, $\ell$, $V'$, $P'$, and $p$, we have that $c(P') \leq \ell$, thus the second condition does not hold (specifically, for each $p \in P' \calR(E)$ we have that $c(p) + c(P' \cap \calR(E)) \leq c(P') \leq \ell$).
Furthermore, the first condition does not hold as $p \notin \calR(E)$.
Hence, contradiction.

\textbf{PR $\to$ Weak-PR}:
Let $\calR$ be an aggregation method satisfying PR. Thus, by definition, for each budgeting scenario $E = (P, V, c, L)$, for each $\ell \in [L]$, each set $V' \subseteq V$ of voters with $|V'| \geq \ell n / L$, and each set $P' \subseteq P$ of projects with $c(P') \leq \ell$, it holds that:
    If all voters $v' \in V'$ support all projects in $P'$, and not any other project,
    then $P' \subseteq \mathcal{R}(E)$.

So, in particular, there exist possible votes for the voters in $V'$ such that $P' \subseteq \mathcal{R}(E)$; e.g., all voters $v' \in V'$ supporting all projects in $P'$, each with $1 / |P'|$ support, and not any other project.
Thus, Weak-PR is satisfied.
\end{proof}

Fain~et~al.~\shortcite{fainMunSha18:core_and_pb} studied another axiom, called \emph{core}, aimed at capturing proportionality for PB. They proved there exists no rule satisfying the core; to the best of our knowledge, our rules are the first that satisfy a natural weakening of the core.

\section{Experimental Evaluation}

\begin{table*}[t]
\caption{Statistics of Simulation Scenario 1 (left),  Simulation Scenario 2 (right) and real data from PB in Warsaw  (bottom).}
\label{table:statistics}
\begin{minipage}{0.49\textwidth}
\centering
\footnotesize
(a) Simulation Scenario 1
\medskip

{
\scriptsize
		\begin{tabular}{ c | c | c | c | c | c} 
			\toprule
			Rule & VS & *VS & suburbs & AR & AC \\
			\midrule
			target & 100\% & 100\% & 40\% & 0.0 & - \\
			\midrule
            GS  &  23.0 \% &  20.0\% & 10.9 \% & 25.4\% & 45k \\ 
            \midrule
            EwT  &  25.0 \% & 23.9\% & 21.6 \% & 2.1\% & 25 k \\
            MT &  22.6 \% & 22.2\% & 31.8 \% & 3.5\% & 24 k \\
            \midrule
            GSC  &  27.6 \% & 25.4\% & 15.2 \% & 11.0\% & 29 k \\
            EwTC & 25.9 \% & 24.1\% & 16.7 \% & 2.8\% & 27 k \\
            MTC  &  25.7 \% & 23.8\% & 22.9 \% & 2.9\% & 28 k \\
     
			\bottomrule
	\end{tabular}
}
\end{minipage}\hfill
\begin{minipage}{0.49\textwidth}
\centering
\footnotesize
(b) Simulation Scenario 2
\medskip

{
\scriptsize
			\begin{tabular}{ c | c | c | c | c | c} 
			\toprule
			Rule & VS & *VS & FoEP &  AR & AC \\
			\midrule
			target & 100\% & 100\% & 50\% & 0.0\% & - \\
			\midrule
            GS  &  21.2 \% & 18.9\% & 65\% & 27.1\% & 45 k \\
			\midrule
            EwT  &  24.4 \% & 23.1\% & 32\% & 12.1\% & 26 k \\
            MT  &  22.8 \% & 22.5\% & 22\% & 19.7\% & 25 k \\
			\midrule
            GSC  &  27.9 \% & 26.3\% & 7\% & 39.0\% & 25 k \\
            EwTC  &  24.5 \% & 22.7\% & 41\% & 8.9\% & 29 k \\
            MTC  &  24.3 \% & 22.9\% & 40\% & 11.2\% & 29 k \\

			\bottomrule
	\end{tabular}
}
\end{minipage}\hfill 
\medskip
\medskip

\begin{minipage}{0.99\textwidth}
\centering
\footnotesize
(c) Warsaw Instance
\medskip

{
\scriptsize
		\begin{tabular}{ c | c | c | c} 
			\toprule
			Rule & VS & AR & AC \\
			\midrule
			WM  &  66\% & 5.1\% & 860 k \\
			\midrule
            GS  &  67\% & 4.6\% & 804 k \\
            \midrule
            EwT  &  80\% & 2.6\% & 295 k \\
            MT  &  80\% & 2.5\% & 294 k \\
            \midrule
            GSC  &  81\% & 2.7\% & 324 k \\
            EwTC  &  81\% & 2.8\% & 319 k \\
            MTC  &  81\% & 2.7\% & 319 k \\
			\bottomrule
	\end{tabular}
}
\end{minipage}
\end{table*}


To complement the axiomatic analysis provided in the section above, in this section we compare our six aggregation methods through computer-based simulations. In particular, we report on simulations done on synthetic datasets as well as on data collected from real-life PB elections.

In our synthetic simulations we randomly generate instances of PB using the \linebreak $2$-dimensional Euclidean model~\cite{enelow1984spatial,mw2d}:
  We associate each voter and each project with a point in a $2$-dimensional Euclidean space; we refer to this point as the \emph{ideal point} of the voter/project. Then, the cumulative ballot of each voter $v_j$ is generated as follows:
  We first identify the set $S_j$ containing the $10$ projects whose ideal points are the closest to the one of $v$.\footnote{As we use real numbers, ties are not an issue.}
Then, we let $v_j$ support only the projects from $S_j$, where, for $p \in S_j$, we set amount by which $v_j$ supports $p$ to be inversely proportional to the distance between the ideal point of $p$ and the ideal point of $v$; specifically: 
  \begin{equation}\label{equation:simulationcosts}
    v_j(p) = \sum_{j: p_i \in S_j} \frac{dist(p_i, v_j(i))}{\sum_{p_k \in S_j} dist(p_k, v_j(i))}\ .
  \end{equation}
  
  We consider two scenarios specifying the distributions of ideal points for voters/projects and of the costs of the projects:

\begin{description}

\item[Simulation Scenario 1:]
This simulation mimics a situation in a municipality, containing a city center and several suburbs around it. The ideal points are drawn as follows:
  Each point has a probability of $0.6$ to fall in a uniform center disc of radius $0.25$ (representing the city center) and a probability of $0.05$ to fall uniformly in one of the eight outer discs of radius $0.06$ each (representing suburbs). 
The cost of each project is selected from a Gaussian with $\mu$ = 50k, $\sigma$ = 20k.

\item[Simulation Scenario 2:]
The purpose of this simulation is to better understand how the considered rules treat projects with different costs. We consider a city containing two centers, one with expensive projects and another one with cheap projects. The ideal points are drawn as follows:
  We have two uniform discs, each with radius 0.25. In the left (right) disc the cost of each project is selected from a Gaussian with $\mu$ = 30k, $\sigma$ = 10k (respectively, $\mu$ = 70k, $\sigma$ = 10k).

\end{description}

\noindent
For each of the two synthetic scenarios, we generate $1000$ PB instances and compute bundle returned by our rules.
We also consider data from the following real instances:

\begin{description}

\item[Warsaw Instance:]
In 2019, 86721 people voted for 101 projects in Warsaw's municipal Participatory Budgeting. Each voter could cast an approval ballot and vote for up to 10 projects. The budget limit was close to PLN 25M ($\approx \$ 6.5$M). The selection rule used there was a simple greedy algorithm, which we refer to as the Warsaw method (WM); WM is equivalent to GS, except that WM takes approval ballots while GS takes cumulative ballots. 
\end{description}


%
In Table~\ref{table:statistics}, we provide the following statistics (for the synthetic data the values are averages over all 1000 instances):


\begin{description}

\item[Voter Satisfaction (VS):]
  fraction of support of a voter which went on funded projects (formally, for a winning bundle $B$, the voter satisfaction of voter $v$ is $\sum_{p \in B} v(p)$).

\item[Anger Ratio (AR):] the fraction of voters who are ignored in the election (formally, $|\{v : \sum_{p \in B} v(p) = 0\}| / |V|$).
  
\item[Average Cost (AC):] average cost of a funded project.

\end{description}

\noindent 
Furthermore, for Scenarios 1 and 2 we compute the $^*$VS metric:

\begin{description}
\item[Voter Satisfaction with Approval Votes ($^*$VS):]
 instead of using Equation~\eqref{equation:simulationcosts}, for each $v_j \in V$ and $p \in S_j$ we set $v_j(p) = 0.1$. This corresponds to using approval ballots (a voter supports each of her supported projects equally).
\end{description}

Finally, for Scenario 1 we include the fraction of the budget spent on projects from the suburbs, 
and for Scenario 2---the fraction spent on projects from the expensive bundle (FoEP).

Our synthetic simulations show that greedy rules produce large anger ratio compared to CSTV rules; also, greedy rules do not satisfy PR (see Table~\ref{table:axioms}), which makes them less attractive for PB. We observe that MT is more fair to smaller districts than EwT and GS, and that MTC is more fair than GSC and EwTC. At the same time, 
we infer that MTC, EwTC, and GS do not discriminate expensive projects as much as the other rules. These results suggest MTC as a particularly good rule, which is only outperformed by MT with respect to how it treats smaller districts. On the other hand, MTC performs considerably better than MT with respect to treating expensive projects, as well as with respect to VS and AR criteria. 

Our simulations based on the Warsaw Instance dataset show that all our methods, with the exception of GS, give better results than WM, the method that is currently used in Warsaw. Bundles selected by our algorithms contribute to larger Voter Satisfaction and lower Anger Ratio. 

\begin{figure}[t]    \label{fig:avg_vs}
    \centering
    \includegraphics[width=7cm]{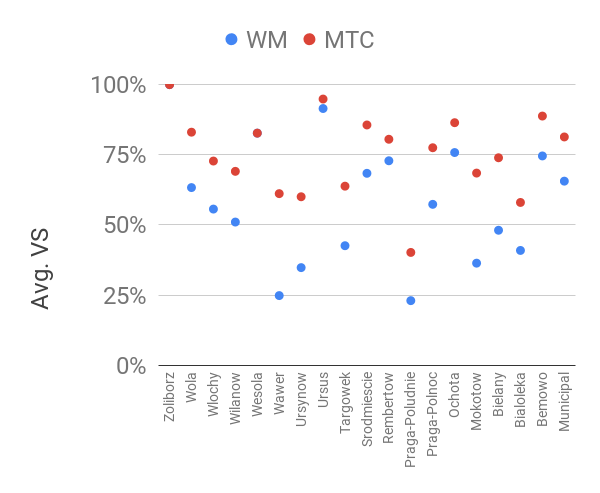}
    \includegraphics[width=7cm]{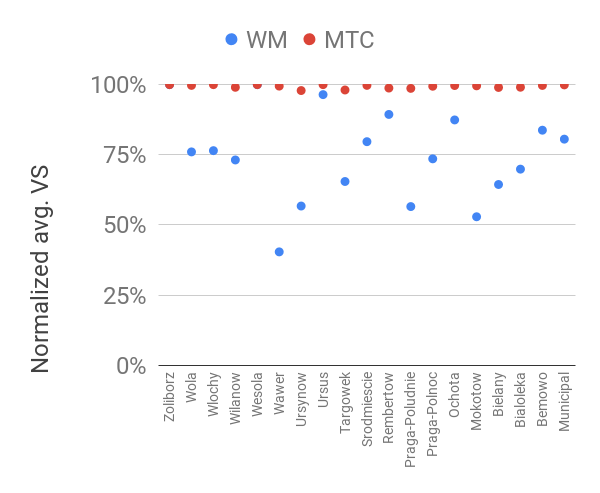}
    \caption{Comparison of WM and MTC. The left (right) picture shows average VS values (average normalized VS values; i.e., percentage of maximum possible VS). On the horizontal axis there are names of 18 local PB instances and one municipal PB instance (the most right one).}
    \label{fig:avg_vs}
\end{figure}

We also compared our methods on 18 additional local PBs in Warsaw. In each of 18 districts there was a separate local PB. The same algorithm was used as in the municipal PB, but each voter could select up to 15 projects. The results are consistent with those presented in Table~\ref{table:statistics}. Moreover, in Figure~\ref{fig:avg_vs} we present the comparison of Warsaw method and MTC. We can see that according to the utilitarian criterion (VS), MTC performs considerably better than the solution that is currently used.

\section{Conclusion}

We suggested the use of cumulative ballots for participatory budgeting and considered several aggregation methods for this setting.
First, our results, in particular the preliminary experimental results presented in Table~\ref{table:statistics}, show that cumulative ballots indeed can improve voter satisfaction.
Moreover, using both our theoretical results and preliminary experimental results, we identified the MTC rule that, we argue, shall be given serious consideration for being used in practical settings, for the following reasons:
\begin{itemize}

\item MTC satisfies Strong-PR, hence is guaranteed to behave in a very proportional way, in particular not dismiss minorities -- in contrast with, e.g., the current method usually used in practice;

\item MTC is computationally efficient, as can be seen by the procedural description of MTC and validated using our simulations;

\item MTC behaves well with respect to our synthetic simulations and in 
real-world data, in particular, it outperforms the method usually used in practice both by increasing the total satisfaction (while still being proportional) as well as by decreasing the frustration rate.

\end{itemize}

We discuss several avenues for future research, regarding applying CSTV in generalized settings.
Indeed, while here we considered the standard combinatorial PB setting, it seems that cumulative ballots and our CSTV rule can be naturally extended to allow for more complex ballots, thus enabling greater voter flexibility that hopefully could result in even better outcomes.
Specifically, in future work we plan to consider CSTV for the following settings:

\paragraph{Negative utilities}
Settings in which voters can express negative utilities -- here, voters could not only state how they wish to split their virtual coin among the projects, but also specify, for each project, whether the support they give to this project is positive or negative.

\paragraph{PB with several resource types}
Settings in which we have not only one type of resource but several, say time and money -- here, voters could get several virtual coins to split simoultaneously;

\paragraph{PB with project interactions}
Settings in which there are interactions between projects, such as in the model of Jain et al.~\cite{jain2020participatory}, assuming a \emph{substitution structure} that is a partition over the projects -- here, in addition to their cumulative ballots, voters could state their preferences regarding substitutions and complementarities between projects; furthermore, as our CSTV rules work by support redistribution on behalf of the voters, we might allow voters to explicitly state how they wish such redistribution to happen.

\section*{Acknowledgements}

Arkadii Slinko was supported by the Faculty Development Research Fund 3719899 of the University of Auckland.
Piotr Skowron was supported by Poland's National Science Center grant UMO-2019/35/B/ST6/02215.
Nimrod Talmon was supported by the Israel Science Foundation (ISF; Grant No. 630/19).

\bibliographystyle{plain}
\bibliography{bib}

\end{document}